\documentclass[conference]{IEEEtran}

\usepackage{setspace}

\usepackage{graphics,
           psfrag,
           epsfig,
           amsthm,
           cite,
           amssymb,
           url,
           dsfont,
           subfigure,
           algorithm,
           algorithmic,
           balance,
           enumerate,
           color,
           setspace
}
\usepackage{amsmath}

\usepackage{epstopdf}

\newtheorem{definition}{Definition}

\newtheorem{theorem}{Theorem}

\newcommand{\eref}[1]{(\ref{#1})}
\newcommand{\sref}[1]{Section~\ref{#1}}

\newcommand{\fref}[1]{Figure~\ref{#1}}

\newcommand{\cref}[1]{Constraint~\ref{#1}}

\newcommand{\ignore}[1]{}

\IEEEoverridecommandlockouts
\begin{document}

\title{A Lossy Graph Model for Decoding Delay Reduction in Instantly Decodable Network Coding}

\author{
   \authorblockN{Ahmed Douik$^{\dagger}$, Sameh Sorour$^\ast$, Mohamed-Slim Alouini$^\dagger$, and Tareq Y. Al-Naffouri$^{\dagger\ast}$\\}%
   \authorblockA{$^\dagger$King Abdullah University of Science and Technology (KAUST), Kingdom of Saudi Arabia \\
    $^\ast$King Fahd University of Petroleum and Minerals (KFUPM), Kingdom of Saudi Arabia \\
    Email: $^\dagger$\{ahmed.douik,slim.alouini,tareq.alnaffouri\}@kaust.edu.sa \\
    $^\ast$\{samehsorour,naffouri\}@kfupm.edu.sa}
    }

\maketitle

\IEEEoverridecommandlockouts

\begin{abstract}

In this paper, we study the broadcast decoding delay performance of generalized instantly decodable network coding (G-IDNC) in the lossy feedback scenario. The problem is formulated as a maximum weight clique problem over the G-IDNC graph in \cite{refsameh}. In order to further minimize the decoding delay, we introduce in this paper the lossy G-IDNC graph (LG-IDNC). Whereas the G-IDNC graph represents only doubtless combinable packets, the LG-IDNC graph represents also uncertain packet combinations when the expected decoding delay of the encoded packet is lower than the individual expected decoding delay of each packet encoded in it. Since the maximum weight clique problem is known to be NP-hard, we use the heuristic introduced in \cite{refjournal} to discover the maximum weight clique in the LG-IDNC graph and finally we compare the decoding delay performance of LG-IDNC and G-IDNC graphs through extensive simulations. Numerical results show that our new LG-IDNC graph formulation outperforms the G-IDNC graph formulation in all situations and achieves significant improvement in the decoding delay especially when the feedback erasure probability is higher than the packet erasure probability. 

\end{abstract}


\section{Introduction} \label{sec:intro}

\emph{Network Coding (NC)} is based on a simple idea. Instead of the simple replication of packets they receive, intermediate nodes in a network may transmit functions of these packets. Since its introduction in \cite{850663}, a decade ago, NC gained much attention thanks to its numerous benefits. It was shown that NC is able to reduce delay over broadcast and multicast erasure channel and thus it can improve the throughput of the channel in the wireless networks. These abilities are of great interest for real time applications requiring fast transmission and recovery of the data such as cellular, streaming television, WiFi, and WiMAX. Two approaches of network coding can be found in the literature named respectively Random Network Coding (RNC) and Opportunistic Network Coding (ONC). Whereas RNC combines packets using random, independent, and non zero coefficients, OPC exploits the diversity of lost and received packets knowledge at the base station (BS) to generate the packet combinations. Despite the attractive merits of RNC such as optimality in reducing the number of sent packets and ability to recover without feedback, it is not suitable for real time applications because decoding can be done only when the whole frame is received.

In this paper, we are interested in a set of applications requiring quick and reliable transmission over lossy channels with strict delay tolerance. A suitable approach of such applications is the OPC subclass called Instantly Decodable Network Coding (IDNC). In this approach, the BS is able to transmit only binary XOR combination of packets. IDNC, despite the limitation in packet generation, is attractive because it allows fast encoding at the BS and fast decoding at the users, eliminating as such the need of expensive computation and reducing also the complexity of both the BS and users. A lot of researches has been conducted to minimize the delay in IDNC. In \cite{ref4}, the authors minimized the completion time, which is the overall transmission time, for IDNC by modeling the problem as a Short Stochastic Path (SSP). The authors in \cite{ref5,ref6,ref1} considered the delay in IDNC as the decoding delay which is the individual delay experienced when delivered a packet that is either non decodable at its arrival or it does not bring new information. The problem was formulated in \cite{ref5} as a linear programming problem called Strict IDNC (S-IDNC), limiting the BS to generate packets that can be decoded by all users then extended to Generalized IDNC (G-IDNC) in \cite{ref2} by loosening this constraint and asking the users instead to discard all non decodable packets. The decoding delay performance of G-IDNC for independent erasure channels was studied in \cite{ref2} and extended to persistent erasure channels in \cite{ref3}.

In the aforementioned works, the authors considered a prompt and perfect reception of the feedback. This assumption is not realistic due to the feedback channel impairments. Recently some research has been dedicated to study the decoding delay with limited feedback \cite{refsameh,refjournal,refahmed,ref17}. More specifically in \cite{refsameh} we studied the decoding delay performance of G-IDNC in lossy and intermittent feedback for memoryless channels. We extended our study to lossy intermittent feedback in \cite{refahmed} and persistent erasure channels in \cite{refjournal}. In these paper, the problem was formulated as a maximum weight clique problem over the G-IDNC graph. Since the problem is known to be NP-hard, many heuristics to perform effective packet selection were proposed \cite{refsameh,refjournal,ref3}.

To represent all the feasible packet combinations, the G-IDNC graph was introduced in \cite{ref2} in a context of perfect feedback. It was shown in \cite{arg1} that all packet combinations is equivalent to a maximal weight clique in the G-IDNC graph and thus the optimal combination that guarantees the minimum decoding delay for the current transmission is the maximum weight clique in the G-IDNC graph. This graph formulation was then used in the lossy feedback context to represent all doubtless packet combination. In others words, the packet combinations that are always instantly decodable for all targeted users. This instantly decodable condition limits the coding opportunities   in each transmission as it discards uncertain packet combination. It clearly limits the minimization of the decoding delay.

In this paper we address the following question : Is there a graph representation that achieves a better decoding delay for G-IDNC in a lossy feedback scenario? To answer this question, we first introduce the Lossy Generalized Instantly Decodable Network Coding graph (LG-IDNC). This new graph allows uncertain packet combinations when the expected decoding delay increase when sending the combined packet is less than the one experienced when sending individually the packets encoded in that combination. Consequently, the LG-IDNC graph is expected to achieve a lower decoding delay than the G-IDNC graph.

The rest of this paper is organized as follows: \sref{sec:model} presents the system, channel, and feedback models. The G-IDNC problem is briefly formulated in \sref{sec:problem} and the G-IDNC graph limitations are exposed in \sref{sec:lim}. We introduce our newly proposed LG-IDNC graph in \sref{sec:lgidnc} before presenting and discussing simulation results in \sref{sec:results}. Finally, \sref{sec:conclusion} concludes the paper.

\section{System and Feedback Model} \label{sec:model}

\subsection{System Model and Parameters}

Consider a wireless base station that is required to broadcast a frame (denoted by $\mathcal{N}$) of $N$ source packets to a set (denoted by $\mathcal{M}$) of $M$ users. Each user is interested in receiving all the packets of $\mathcal{N}$ with the minimum delay in any order. In the \emph{initial phase}, the BS broadcasts the $N$ packets of the frame uncoded to all the users. Each user who received a packet acknowledges its reception by sending a feedback. A lost packet or feedback will make the state of the user/packet uncertain at the BS.

After the \emph{initial phase}, three sets of packets are attributed to each user $i$:

\begin{itemize}
\item The Has set (denoted by $\mathcal{H}_i$) is defined as the set of packets successfully received and acknowledged by user $i$.
\item The Wants set (denoted by $\mathcal{W}_i$) is defined as the set of packets that are not in the Has set ( $\mathcal{W}_i = \mathcal{N} \setminus \mathcal{H}_i$ ). In other words, $\mathcal{W}_i$ consists of the lost or failed to feedback packets by user $i$.
\item The Uncertain set (denoted by $\mathcal{U}_i$) is defined as the set of packets which state is uncertain. We have $\mathcal{U}_i \subseteq \mathcal{W}_i$.
\end{itemize}

The BS stocks these information in a \emph{BS feedback matrix} (SFM) $\mathbf{F} = [f_{ij}],~ \forall~ i \in \mathcal{M},~ \forall~j \in\mathcal{N}$ as follows:
\begin{align}
f_{ij} =
\begin{cases}
0 \hspace{0.9 cm}& j \in \mathcal{H}_i \\
1 \hspace{0.9 cm}& j \in \mathcal{W}_i \setminus \mathcal{U}_i \\
x \hspace{0.9 cm}& j \in \mathcal{U}_i.
\end{cases}
\end{align}

After this \emph{initial phase}, the recovery phase takes place. In this phase, the BS uses the SFM to select the network XOR-coded combinations of the source packets to transmit. Each user that received and successfully decoded a packet sends back an acknowledgement after each transmission. The BS uses this feedback to update the feedback matrix. This process is repeated until all users acknowledge the successful reception of all the packets.

The encoded packets, in each transmission, can have one of the following options for each user $i$:
\begin{itemize}
\item \emph{Non-innovative:} A packet is non-innovative for user $i$ if the packets encoded in it do not bring new information. In other words, they were all previously received and successfully decoded.
\item \emph{Instantly Decodable:} A packet is instantly decodable for user $i$ if it contains \emph{only one source packet} from $\mathcal{W}_i$.
\item \emph{Non-Instantly Decodable:} A packet is non instantly decodable for user $i$ if it contains two or more source packet from $\mathcal{W}_i$.
\end{itemize}

We define the decoding delay as in \cite{ref2}:
\begin{definition}
At any recovery phase transmission, a user $i$, with non-empty Wants sets, experiences a one unit increase of decoding delay if it receives a packet that is either non-instantly decodable or both instantly decodable and non-innovative.
\end{definition}

We define the targeted users by a transmission as the set users to whom the BS indented the packet combination when encoding it.

\subsection{Channel and Feedback Models}

The channel is modeled as a memoryless erasure channel. Each packet is subject to loss at user $i$ with a packet erasure probability $p_i,~\forall~ i \in \mathcal{M}$. We also assume that the channels of various users are independent.

The feedback channels follow a similar model: they are independent from one user to the other and the $i$th channel experiences erasure with probability $q_i ,~ \forall~ i \in \mathcal{M}$. The feedback consists of an acknowledgement of all previously received/lost packets. Only targeted users feedback. Thus the BS will not receive any feedback from a user unless it was targeted by the transmission.

The packet and feedback erasure probabilities are constant during a frame delivery period and are known by the BS. The special case when the packet and feedback erasure probabilities are the same for all users is called reciprocal channel. 

\section{Generalized Instantly Decodable Network Coding Problem (G-IDNC)} \label{sec:problem}

\subsection{Problem Description and Formulation}

For an arbitrary packet combination $\kappa$, let $d_i(\kappa)$ denote the decoding delay increase of user $i$ and let $\mathcal{D}\left(\kappa\right)$ be the total increase of the decoding delay of all users after the transmission of the encoded packet $\kappa$. In other words, we have:
\begin{align}
\mathcal{D}\left(\kappa\right) = \sum_{i \in M_w} d_i(\kappa),
\end{align}
where $M_w$ is the set of users having non empty Wants set.

Minimizing the decoding delay is finding the packet combination among all possible packet combinations that guarantees the minimum expected decoding delay for the current transmission. The problem was formulated in \cite{ref5} as follows:
\begin{align}
\underset{\kappa}{\text{min}}\left\{ \mathds{E} \left[ \mathcal{D}\left(\kappa\right) \right]\right\} = \underset{\kappa}{\text{min}}&\left\{ \mathds{E} \left[\sum_{i \in M_w} d_i(\kappa)\right]\right\} \nonumber \\
\text{subject to }\kappa &= \bigoplus_{j \in A} j,~\forall~A \subseteq \mathcal{N},
\end{align}
where $\bigoplus\limits_{j \in A}$ represents the XOR combination of packets in the set $A$.

In order to minimize the decoding delay in G-IDNC in the perfect feedback case, the authors in \cite{ref2} proposed to look for all doubtless possible packet XOR-combinations then choose the combination that guarantees the minimum delay for this transmission.

\subsection{G-IDNC Graph}

To exhibit all these possible combinations of packets, the \emph{G-IDNC graph} was introduced in \cite{ref2}. G-IDNC considers that two packets can be combined if the packet that is in the Wants set of the first user is in the Has of the second user and inversely. This approach was extended in \cite{refsameh,refjournal,refahmed} to the lossy feedback. However in this configuration G-IDNC graph presents only all the doubtless possible packet XOR-combinations.

This $\mathcal{G} (\mathcal{V},\mathcal{E})$ graph is constructed by first generating a vertex $v_{ij} \in \mathcal{V}$ for each nonzero entry in the SFM (i.e. for every packet $j \in \mathcal{W}_i,~ \forall~ i \in \mathcal{M}$). These vertices are then connected with an edge if the packet combination represented by the vertices can be decoded by the two users represented by these vertices. In other words, $v_{ij}$ and $v_{kl}$ in $\mathcal{V}$ are connected with an edge if and only if one of the following conditions is true:
\begin{itemize}
\item C1: $j=l \Rightarrow$ The users $i$ and $k$ are interested in receiving the same packet $j$.
\item C2: $j \in \mathcal{H}_k$ and $l \in \mathcal{H}_i \Rightarrow$ The requested packet of the first vertex is in the Has set of the second vertex and reciprocally.
\end{itemize}

The connectivity conditions C1 expresses the interest of two users in the same packet and the condition C2 involves combination of packet $j$ and $l$ that will be certainty instantly decodable and maybe non-innovative for users $i$ and $k$. 

The BS encodes the packet to be sent by taking the binary XOR of all the packets represented by the vertices of a selected maximal clique in $\mathcal{G} (\mathcal{V},\mathcal{E})$. The users to whom this encoded packet is intended are those identified by the vertices of the selected maximal clique.

In \cite{arg1}, it was shown that the problem of minimizing the decoding delay in G-IDNC is equivalent to finding the maximum weight clique in the G-IDNC graph. The maximum weight clique problem in lossy feedback was formulated in \cite{refjournal} as follows:
\begin{align}
\kappa^* &=  \underset{\kappa \in \mathcal{G}}{\text{argmax}}  \sum_{i \in \tau(\kappa)} (1-p_i) \times p_{i,n}(j_i(\kappa)),
\end{align}
where $\tau(\kappa)$ is the set of users targeted by $\kappa$, $j_i(\kappa)$ is the targeted packet in the transmission $\kappa$ to user $i$ and $p_{i,n}(j)$ the probability that packet $j$ is innovative for user $i$ given by the following theorem:

\begin{theorem}
The probability for packet $j$ to be innovative to user $i$ is
\begin{align}
p_{i,n}(j) &= \left(\cfrac{p_i}{p_i+(1-p_i)q_i}\right)^{\lambda_{ij}},~\forall~j \in \mathcal{W}_i,
\end{align}
where $\lambda_{ij}$ is the number of time packet $j$ was attempted to receiver $i$ since the last heard feedback from that receiver.
\end{theorem}

\begin{proof}
The proof of this theorem can be found in \cite{refsameh}.
\end{proof}

\section{G-IDNC Graph Limitations}\label{sec:lim}

From the connectivity conditions of G-IDNC graph, we clearly see that this graph does not represent all the possible packet combinations but only those that are certainly instantly decodable by all the targeted users. Therefore, in the lossy feedback scenario this graph represents a sub-optimal solution. Consider the following BS feedback matrix and the associated G-IDNC graph illustrated in \fref{fig:example}. We assume in this example that all users are experiencing the same packet erasure probability $p$.
\begin{figure}[t]
\centering
  \includegraphics[width=1\linewidth]{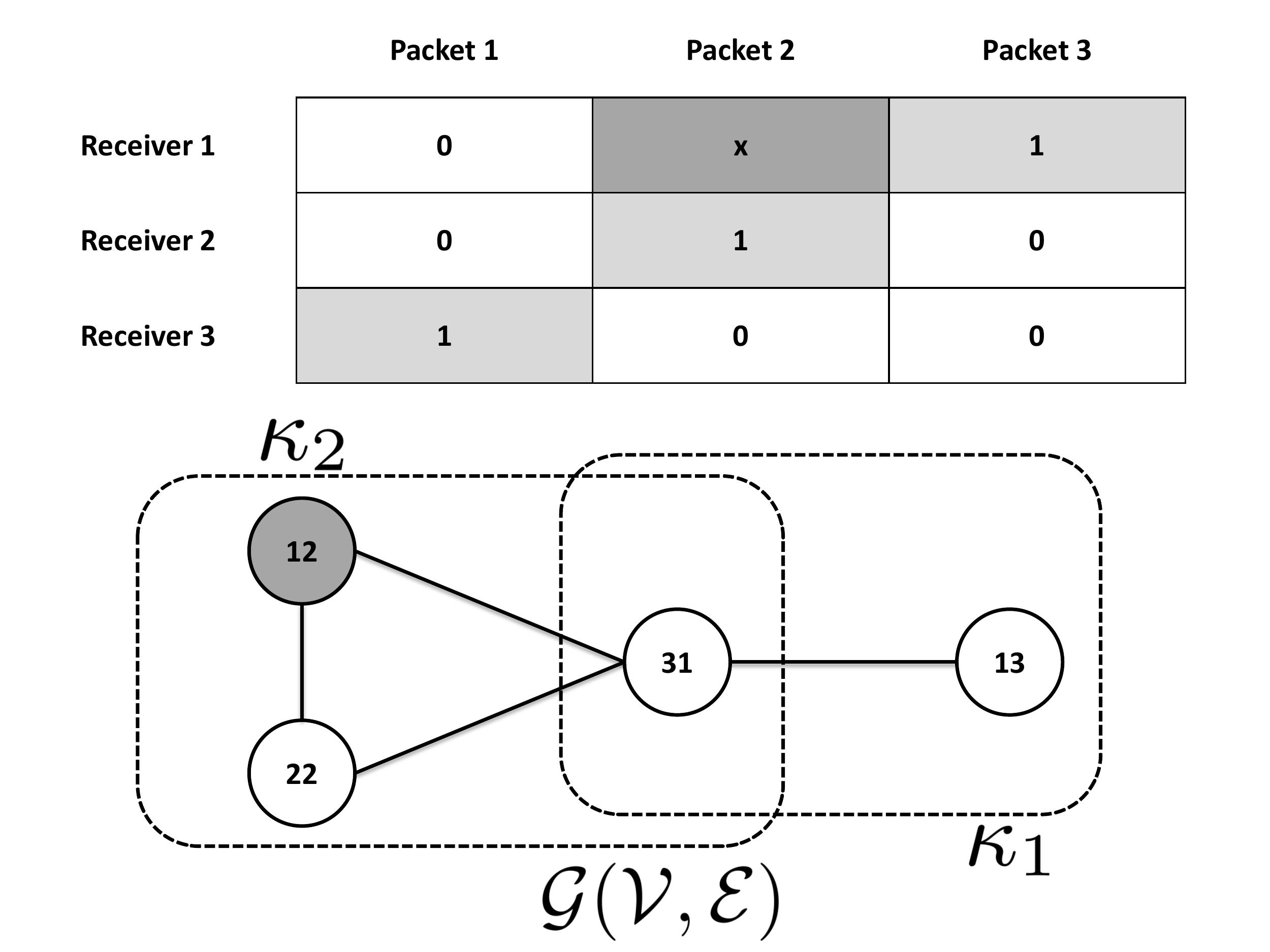}\\
  \caption{Feedback matrix and its corresponding G-IDNC graph.}\label{fig:example}
\end{figure}

The possible packet combinations in the G-IDNC graph, represented by the maximal weight cliques, are $1\oplus2$ and $1\oplus3$. The expected decoding delay increase after sending $1\oplus2$ is $(1-p)p_{n,1}(2)$ and the one after sending $1\oplus3$ is $(1-p)$. On the other hand, the coded packet $1\oplus2\oplus3$, which violates the G-IDNC graph conditions, achieves an expected decoding delay increase of $(1-p)(1-p_{n,1}(2))$. Consequently, if $(1-p_{n,1}(2)) < p_{n,1}(2)$, the packet $1\oplus2\oplus3$ is expected to achieve a lower decoding delay. However this packet belongs to another class of graph with different connectivity conditions. After this motivating example, we now describe our newly proposed lossy G-IDNC graph that can be used to further minimize the sum decoding delay.

\section{Lossy Generalized Instantly Decodable Network Coding Graph (LG-IDNC)} \label{sec:lgidnc}

In this section, we introduce our proposed graph formulation. In order to express the connectivity conditions of the LG-IDNC graph, we need first to compute the expected decoding delay increment of an encoded packet. We finally formulate the graph connectivity conditions based on the expected decoding delay increase. 

\subsection{Decoding Delay Increment of a Packet Combination}

Let $v_{ij}$ and $v_{kl}$ be two vertices in G-IDNC graph with $i \neq k$ and $j \neq l$. Note that packets in the Has set can be considered as uncertain packets with a probability to be innovative equal to $0$ and certain packet in the Wants set (i.e. $\mathcal{W}_i \setminus \mathcal{U}_i$) as uncertain packets with innovative probability equal to $1$. Therefore, for any user $i$ and for any arbitrary packet $j$, the probability that this packet bring new information is $\widehat{p}_{i,n}(j)$ with
\begin{align}
\widehat{p}_{i,n}(j) = 
\begin{cases}
p_{i,n}(j) \hspace{0.9 cm}& j \in \mathcal{U}_i\\
0 \hspace{0.9 cm}& j \in \mathcal{H}_i\\
1 \hspace{0.9 cm}& j \in \mathcal{W}_i \setminus \mathcal{U}_i.
\end{cases}
\end{align}
The following theorem introduces the finish probability:

\begin{theorem}
The probability that user $i$ successfully received all his primary packets but $\mathcal{W}_i \neq \varnothing$ at time $t$ is:
\begin{align}
p_{i,f}  &= \prod_{j \in \mathcal{W}_i} \left( 1- p_{i,n}(j) \right) , ~ \forall~ i \in \mathcal{M} \text{ such that } \mathcal{W}_i = \mathcal{U}_i.
\end{align}
\end{theorem}
\begin{proof}
The proof of this theorem can be found in \cite{refsameh}.
\end{proof}

Note that if user $i$ received all the needed packets, we can consider that $p_{i,f}=1$ and if he has packets in $\mathcal{W}_i \setminus \mathcal{U}_i$, we can consider that $p_{i,f}=0$ for that user. Therefore, for any arbitrary user $i$, the probability that this receiver still need packets can be written as:
\begin{align}
\overline{p}_{i,f}  &= 1 - \widehat{p}_{i,f} = 1 - \prod_{j \in \mathcal{N}} \left( 1- \widehat{p}_{i,n}(j) \right).
\end{align}

Let $\mathcal{D}(j)$ denote the overall expected decoding delay increase for users $i$ and $k$ after sending packet $j$. In other words:
\begin{align}
\mathcal{D}(j) = \mathds{P}(d_i(j)=1) + \mathds{P}(d_k(j)=1).
\end{align}
When only packet $j$ is sent, user $i$ will experience a delay if the two following conditions are true:

\begin{enumerate}
\item He receivers the packet $j$.
\item The packets $j$ is not innovative for that user.
\item He still needs packets.
\end{enumerate}
Thus the decoding delay increase for this user is 

\begin{align}
\mathds{P}(d_i(j)=1) = (1-p_i)(1-\widehat{p}_{n,i}(j))\overline{p}_{f,i} .
\end{align}
By symmetry user $k$ will also experience a similar delay. The sum decoding delay $\mathcal{D}(j)$ experienced by both users $i$ and $k$ when sending packet $j$ can be expressed as: 
\begin{align}
\label{eq1}
\mathcal{D}&(j) = d_{ij,kl}(j)  \\
&= (1-p_i)(1-\widehat{p}_{n,i}(j))\overline{p}_{f,i} + (1-p_k)(1-\widehat{p}_{n,k}(j))\overline{p}_{f,k}.\nonumber
\end{align}
Similarly, when sending packet $l$ for users $i$ and $k$ the decoding delay $\mathcal{D}(l)$ can be obtained by replacing $j$ by $l$ in \eref{eq1}.

If the encoded packet is $j \oplus l$, user $i$ will experience one unit of decoding delay if the following conditions are true:

\begin{enumerate}
\item He receivers the packet $j \oplus l$.
\item The packets $j $ and $ l$ are either in this Has set or in his Wants set. The probability of this event is 
\begin{align*}
\widehat{p}_{n,i}(j)\widehat{p}_{n,i}(l)+ (1-\widehat{p}_{n,i}(j)(1-\widehat{p}_{n,i}(l)).
\end{align*}
\item He still needs packets.
\end{enumerate}

User $k$ will experience a similar delay when sending the encoded packet $j \oplus l$. Therefore, the overall decoding delay $\mathcal{D}(j \oplus l)$ for users $i$ and $k$ is given by:
\begin{align}
\label{joplusl}
& \hspace{1cm} \mathcal{D}(j \oplus l) = d_{ij,kl}(j \oplus l) = \\
&  (1-p_i)(\widehat{p}_{n,i}(j)\widehat{p}_{n,i}(l)+ (1-\widehat{p}_{n,i}(j)(1-\widehat{p}_{n,i}(l)))\overline{p}_{f,i}+  \nonumber \\
& (1-p_k)(\widehat{p}_{n,k}(j)\widehat{p}_{n,k}(l)+ (1-\widehat{p}_{n,k}(j)(1-\widehat{p}_{n,k}(l)))\overline{p}_{f,k}.\nonumber
\end{align}

Note that, since $j \oplus 0 = j$, then $d_{ij,kl}(j)$ can be obtained by replacing $l$ by $0$ in \eref{joplusl} and taking $\widehat{p}_{n,i}(0) =0$ (i.e. $0$ is not an innovative packet for every user).

\subsection{LG-IDNC Graph Construction}

To further minimize the decoding delay for G-IDNC in a lossy feedback scenario, we look for all possible packet combinations then select the one that guarantees a minimum decoding delay increase for the current transmission. To represent all these potential packet combinations, we use a graph model similar to the G-IDNC graph introduced in \cite{ref2}. In our lossy feedback context, we will call this graph the LG-IDNC.

The LG-IDNC graph is constructed like the G-IDNC graph by first generating a vertex for all nonzero entries in the SFM. We then connect these vertices if the expected decoding delay of the coded packet is lower than the decoding delay experienced when sending the packets individually. In other word, two vertices $v_{ij}$ and $v_{kl}$ are connected by an edge if one of the following conditions is true:
\begin{itemize}
\item C1: $j=l \Rightarrow$ Packet $j$ is needed by both the users $i$ and $k$.
\item C2: $d_{ij,kl}(j \oplus l) \leq min (d_{ij,kl}(j) ,d_{ij,kl}(l) ) \Rightarrow$ The packet combination $j \oplus l$ guarantees a lower decoding delay than individuals packets.
\end{itemize}

Giving these connectivity conditions, we can clearly see that C2 includes the second condition of the G-IDNC graph by taking $p_{n,i}(l)=p_{n,k}(j)=0$. This proves that the G-IDNC graph is a subgraph of the LG-IDNC graph and so it is guaranteed to achieve lower delay. The set of all feasible encoding packets is defined by all maximal cliques in the LG-IDNC graph. 

Note that, given the innovative and the finish probabilities for every couple user/packet, our newly graph formulation is ensured to outperform the G-IDNC graph for all feedback scenario. The expressions of the innovative and the finish probabilities for multiple feedback imperfections scenarios can be found in \cite{refjournal}. In the special case when the feedback is perfect, the G-IDNC graph and the LG-IDNC graph are the same graph.

To solve this NP-hard maximum weight clique problem, we use the heuristic introduced in \cite{refjournal} to perform effective packets selection. The set of users to whom the encoded packet is intended are those represented by the maximal clique and thus they are the ones that will send feedback upon successful reception of the packet combination regardless if they decoded it or not.

\section{Simulation Results} \label{sec:results}

\begin{figure}[t]
\centering
  \includegraphics[width=1\linewidth]{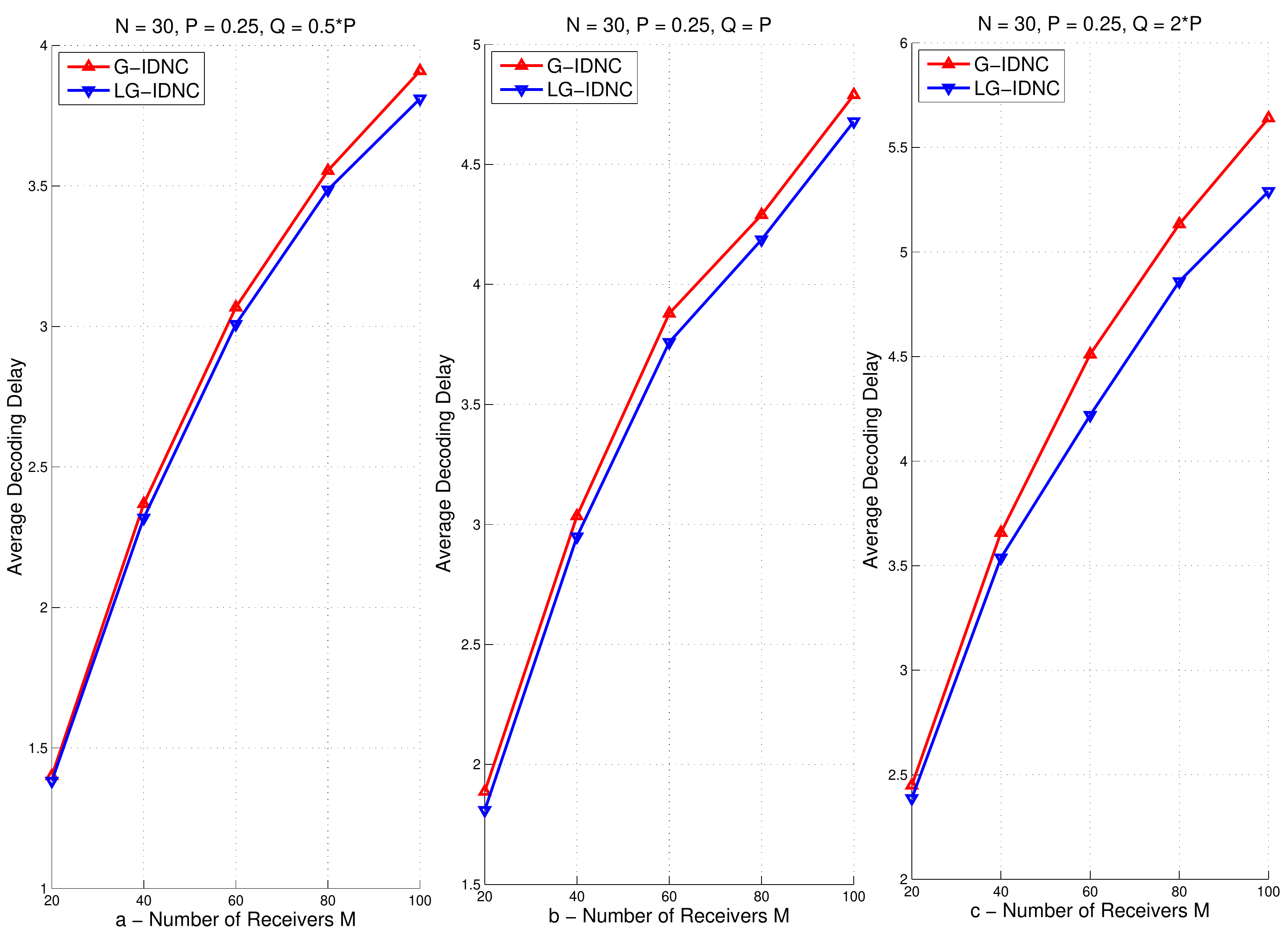}\\
  \caption{Mean decoding delays against $M$ for a low erasure channel.}\label{fig:LOSSYM}
\end{figure}

In this section, we compare the performance achieved by our lossy generalized IDNC graph (denoted by LG-IDNC) against the one achieved by G-IDNC graph (denoted by G-IDNC) to reduce the decoding delay in G-IDNC over lossy feedback. The packet erasure probability ($p_i, ~\forall~ i \in \mathcal{M}$) and the feedback erasure probability ($q_i, ~\forall~ i \in \mathcal{M}$) for all users are considered constant during a delivery frame and they change uniformly in a given range from iteration to iteration while keeping its mean $P$ and $Q$, respectively, constant for all the simulations. The decoding delay is computed over a large number of iterations and the average value is presented. We have $0<P,Q<0.8$.

\begin{figure}[t]
\centering
  \includegraphics[width=1\linewidth]{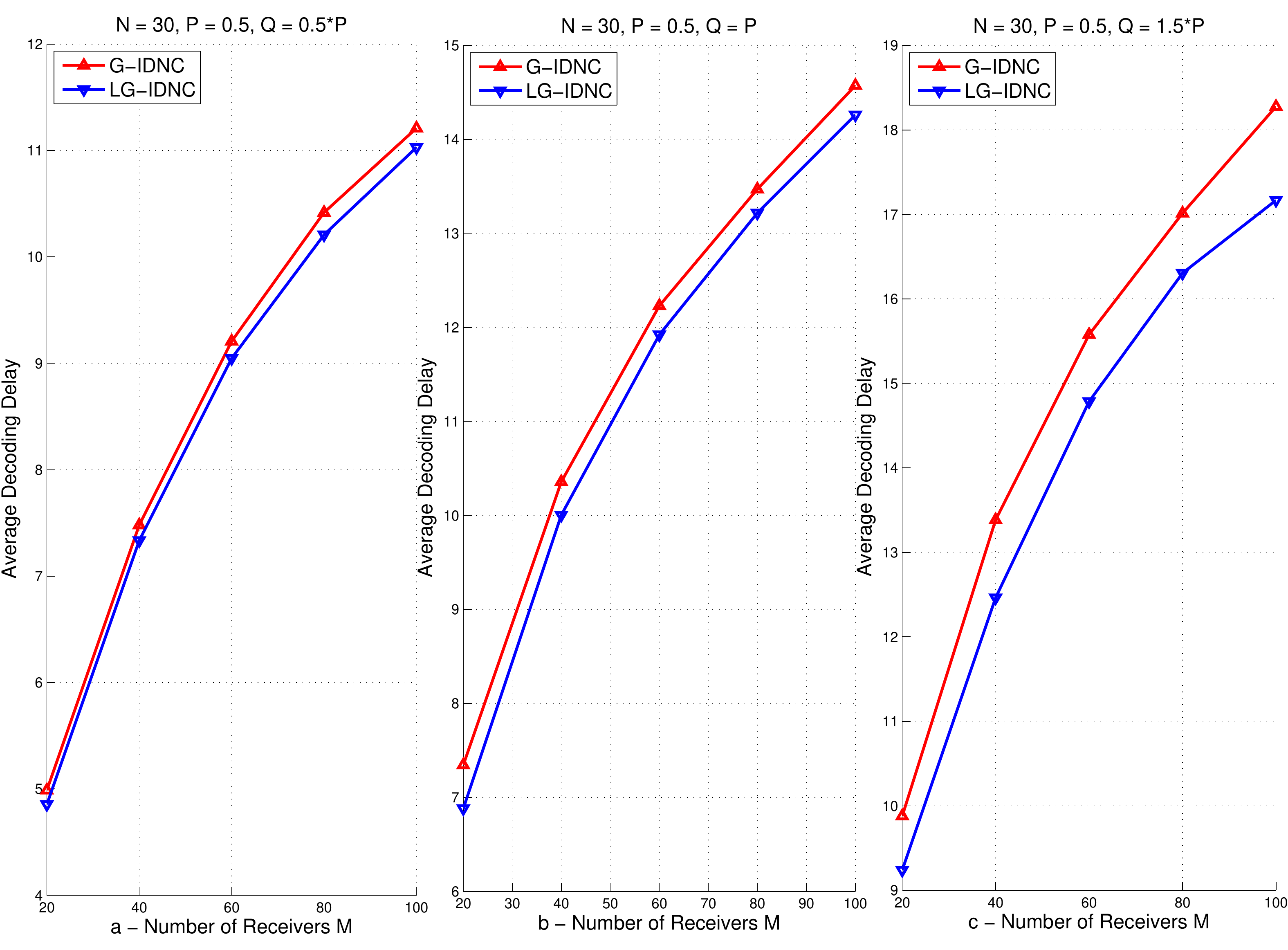}\\
  \caption{Mean decoding delays against $M$ for a high erasure channel.}\label{fig:LOSSYM2}
\end{figure}
\begin{figure}[t]
\centering
  \includegraphics[width=1\linewidth]{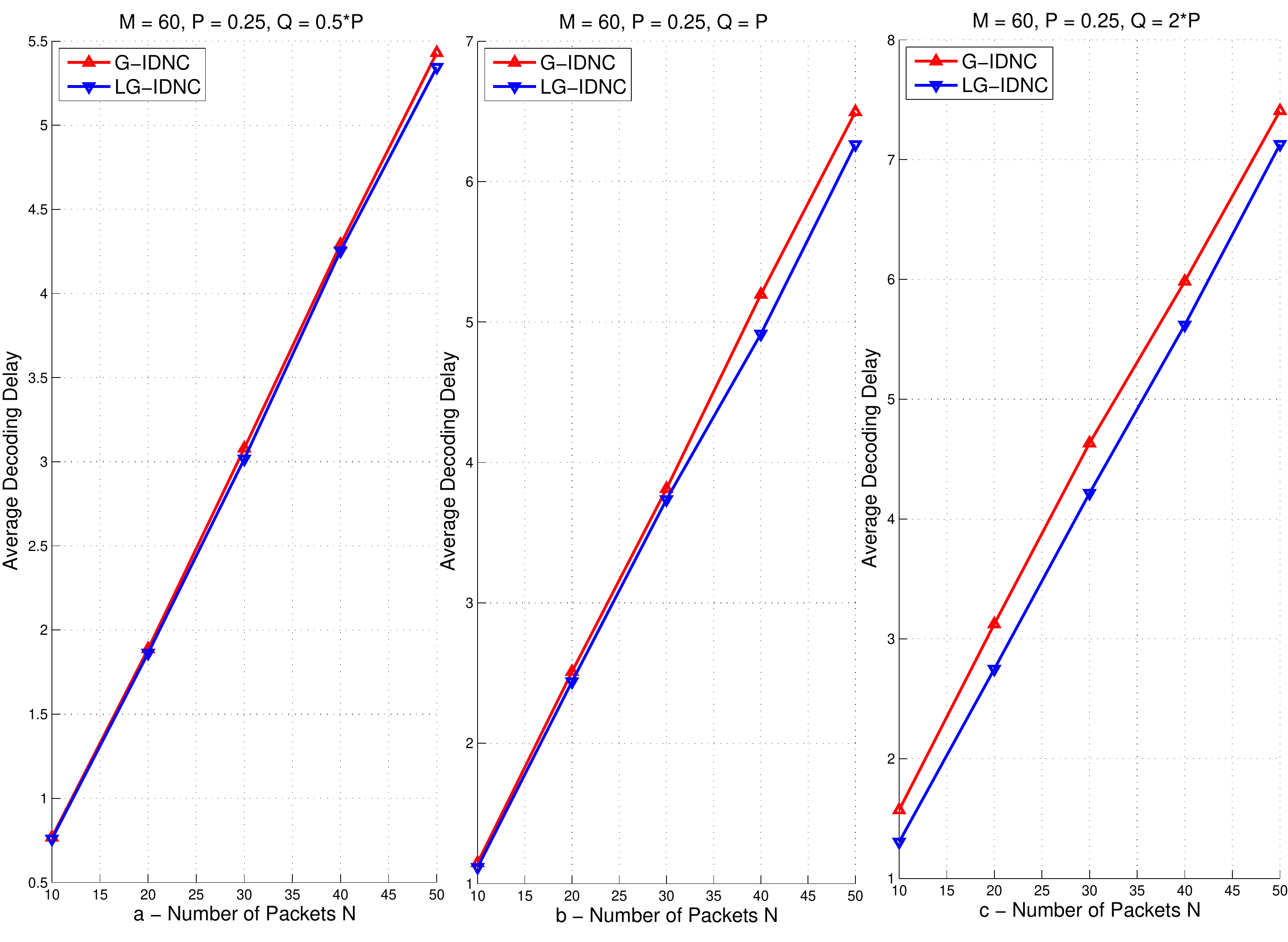}\\
  \caption{Mean decoding delays against $N$ for a low erasure channel.}\label{fig:LOSSYN}
\end{figure}
\begin{figure}[t]
\centering
  \includegraphics[width=1\linewidth]{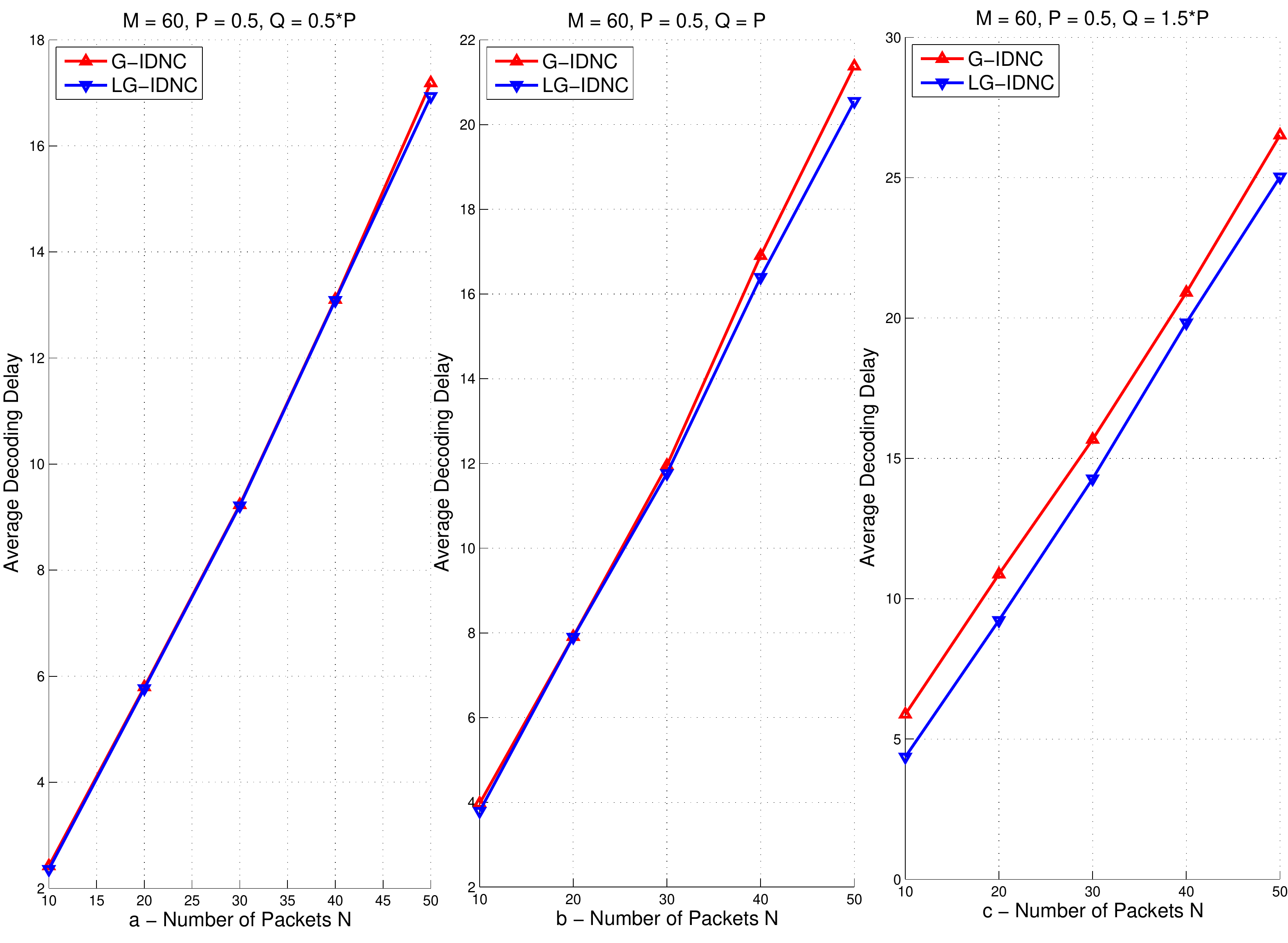}\\
  \caption{Mean decoding delays for intermittent lossy feedback versus $N$.}\label{fig:LOSSYN2}
\end{figure}
\begin{figure}[t]
\centering
  \includegraphics[width=1\linewidth]{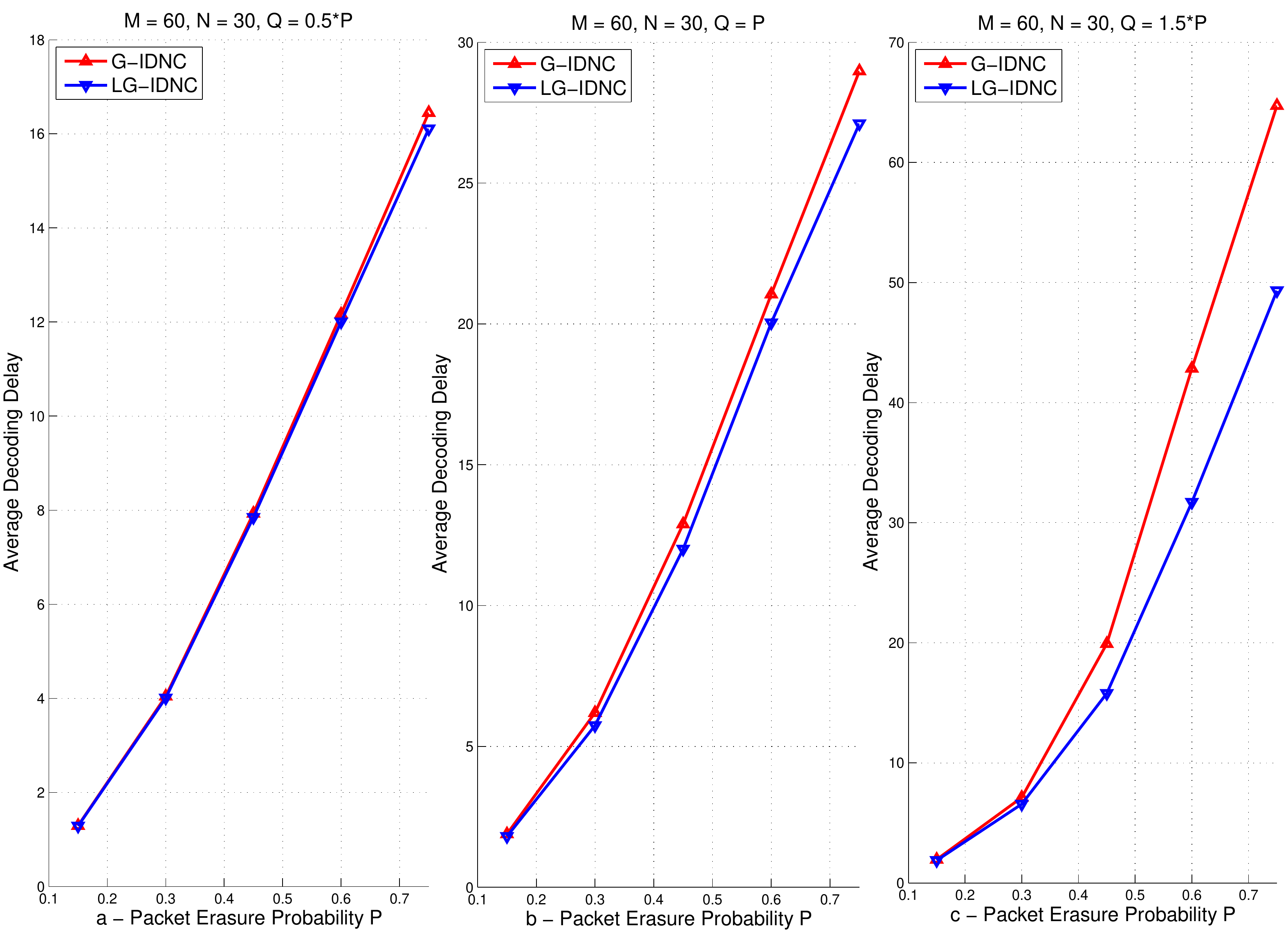}\\
  \caption{Mean decoding delays for lossy feedback versus $P$.}\label{fig:LOSSYP}
\end{figure}

\fref{fig:LOSSYM} depicts the comparison of mean decoding delays achieved by LG-IDNC and G-IDNC graphs against $M$, for $N = 30$, $P=0.25$ and, respectively, for $Q=\cfrac{P}{2}$, $Q=P$ and $Q=2P$. \fref{fig:LOSSYM2} depict the same comparison for a higher packet erasure probability $P=0.5$ and $Q=\cfrac{3P}{2}$. \fref{fig:LOSSYN} and \fref{fig:LOSSYN2} illustrate the comparison of the performance of the graphs against $N$, for $M = 60$, $Q=\cfrac{P}{2}$, $Q=P$ and $Q=2P$ for a low ($P=0.25$) and a high ($P=0.5$, $Q=\cfrac{3P}{2}$) erasure channel respectively. \fref{fig:LOSSYP} shows the same comparison against the packet erasure probability for $M=60$, $N=30$ and and, respectively, for $Q=\cfrac{P}{2}$, $Q=P$ and $Q=\cfrac{3P}{2}$.

From all the figures, we can see that our proposed graph formulation achieves a better decoding delay in all the situation. The average decoding delay gain from \fref{fig:LOSSYM} and \fref{fig:LOSSYN} when the channel is reciprocal is $3\%$ and when the feedback erasure probability is higher than the packet erasure probability is $7\%$. \fref{fig:LOSSYM2}.c and \fref{fig:LOSSYN2}.c shows that our LG-IDNC graph offers a gain of $9\%$ when the channel conditions are harsher for high feedback erasure probability.

\fref{fig:LOSSYM}.a, \fref{fig:LOSSYN}.a, \fref{fig:LOSSYM2}.a and \fref{fig:LOSSYN2}.a show that LG-IDNC and G-IDNC achieve a close decoding delay for a low feedback erasure probability. This can be explained by the fact that in low feedback erasure probability scenario, the probability to loose the transmission and the probability to loose the feedback are close which make packet state estimation non effective for our graph formulation. However, when the channel erasure is lower or equal to the feedback erasure probability, the estimation is more accurate which explain the difference between the achieved decoding delay by LG-IDNC and G-IDNC.

From \fref{fig:LOSSYP}, we clearly can see the gap between our proposed graph formulation and G-IDNC when the persistence of the channel is higher than $0.3$. This can be explained by the fact that when the erasure of the channel increases, the probability for an uncertain packet increase also. Therefore, more vertices are likely to be connected in the graph and as a consequence the encoded packet will targeted more users achieving a lower decoding delay.

\section{Conclusion} \label{sec:conclusion}

In this paper, we first introduced the LG-IDNC graph to further minimize the broadcast decoding delay of generalized instantly decodable network coding in the lossy feedback scenario compared to the G-IDNC graph. While G-IDNC graph limits the packet selection to only certainty instantly decodable combinable packets to all targeted users, the LG-IDNC graph represents all these combinations under the condition that the expected decoding delay of the encoded packet is lower than the individual expected decoding delay of each packet encoded on it. Through extensive simulations, we showed that the decoding delay performance of the LG-IDNC graph are better than the one achieved by the G-IDNC graph and more significantly when the channel is not reciprocal. 

\appendices

\bibliographystyle{IEEEtran}
\bibliography{references}

\end{document}